\documentclass[preprint,12pt]{elsarticle}

\usepackage[utf8]{inputenc}
\usepackage[T1]{fontenc}
\usepackage{graphicx}
\usepackage{hyperref}
\usepackage{amsmath,amssymb,amsfonts,amsthm,mathtools}

\theoremstyle{plain}
\newtheorem{theorem}{Theorem}[section]

\newtheorem{lemma}[theorem]{Lemma}

\theoremstyle{definition}
\newtheorem{definition}[theorem]{Definition}
\newtheorem{assumption}[theorem]{Assumption}

\theoremstyle{remark}

\usepackage{algorithm}
\usepackage{algpseudocode}
\usepackage{caption,subcaption,wrapfig,array,booktabs,siunitx}

\makeatletter
\def\ps@pprintTitle{%
  \let\@oddhead\@empty
  \let\@evenhead\@empty
  \let\@oddfoot\@empty
  \let\@evenfoot\@empty}
\makeatother

\journal{Journal of Computational Science}

\begin{document}
\begin{frontmatter}

\title{Velocity-Inferred Hamiltonian Neural Networks:\\
       Learning Energy-Conserving Dynamics from Position-Only Data}

\author[addr1]{Ruichen Xu}

\author[addr3]{Zongyu Wu}

\author[addr2]{Luoyao Chen}

\author[addr1]{Georgios Kementzidis}

\author[addr4]{Siyao Wang}

\author[addr1]{Haochun Wang}

\author[addr1]{Yiwei Shi}

\author[addr1]{Yuefan Deng\corref{cor1}}

\cortext[cor1]{Corresponding author}

\affiliation[addr1]{%
  organization={Department of Applied Mathematics and Statistics, Stony Brook University}, country={USA}}

\affiliation[addr2]{%
  organization={New York University}, country={USA}}

\affiliation[addr3]{%
  organization={Independent Researcher}, country={USA}}

\affiliation[addr4]{%
  organization={Department of Statistics, University of California, Davis},country={USA}}

\begin{abstract}
Data-driven modeling of physical systems often relies on learning both positions and momenta to accurately capture Hamiltonian dynamics. However, in many practical scenarios, only position measurements are readily available. In this work, we introduce a method to train a standard Hamiltonian Neural Network~(HNN) using \emph{only} position data, enabled by a theoretical result that permits transforming the Hamiltonian $H(q,p)$ into a form $H(q,v)$. Under certain assumptions—namely, an invertible relationship between momentum and velocity—we formally prove the validity of this substitution and demonstrate how it allows us to infer momentum from position alone. We apply our approach to canonical examples including the spring–mass system, pendulum, two-body, and three-body problems. Our results show that using only position data is sufficient for stable and energy-consistent long-term predictions, suggesting a promising pathway for data-driven discovery of Hamiltonian systems when momentum measurements are unavailable.
\end{abstract}


\end{frontmatter}

\begin{center}
\footnotesize
\textbf{Note:} This manuscript is a \emph{preprint (ongoing work)}. 
We may update or revise it significantly before official publication.
\normalsize
\end{center}

\section{Introduction}
\label{sec:intro}

Modeling high-fidelity dynamics in physical systems is a longstanding challenge in both scientific computing and artificial intelligence.  
Classical methods often rely on specialized numerical schemes, such as symplectic integrators \cite{Hairer2006}, to preserve geometric properties of the underlying equations.  
However, these methods can become cumbersome when the system operates in high-dimensional spaces or when only partial observations (such as positions) are available \cite{Celledoni2021,Azencot2022}.  
In recent years, data-driven approaches have emerged that discover governing laws directly from observations while maintaining critical physical constraints, exemplified by neural ordinary differential equation frameworks \cite{Chen2018} and physics-informed neural networks \cite{Raissi2019}.

Among these advances, \emph{Hamiltonian Neural Networks (HNNs)} learn a Hamiltonian function whose partial derivatives define energy-conserving dynamics \cite{Greydanus2019}.  
Traditionally, such models assume that both positions and canonical momenta are accessible, allowing the Hamiltonian to be trained from direct measurements of how states evolve over time.  
In many practical scenarios, however, momentum measurements are unavailable, and only position data can be collected.  
This limitation arises in domains such as orbital tracking or biological motion analysis, where sensors record spatial trajectories but cannot measure the physical impulse driving those trajectories.

\paragraph{Canonical momentum vs.\ velocity}
The momentum variable, often denoted \(p\) in Hamiltonian systems, is not always a simple “mass times velocity.”  
In Cartesian coordinates with constant mass, this may hold (i.e., \(p = m\,v\)), but in more general coordinate systems (such as a pendulum’s angular coordinates) or systems with varying mass matrices, the relationship between \(p\) and velocity \(v\) can be more nuanced.  
Nevertheless, under broad conditions, an invertible mapping \(p = M(q)\,v\) (for some mass-like matrix \(M\)) allows one to recover the velocity from measured positions (via finite differencing) and treat it as a surrogate for momentum.

\paragraph{Introducing Velocity-Inferred Hamiltonian Neural Networks (VI-HNN)}
We propose a straightforward strategy to learn Hamiltonian dynamics \emph{without direct momentum observations}, relying on the assumption that momentum can be approximated by velocity through an invertible mapping.  
In this approach, position data alone are used to estimate velocities, which in turn serve as substitutes for canonical momenta in the Hamiltonian formulation.  
We refer to this architecture as the \emph{Velocity-Inferred Hamiltonian Neural Network (VI-HNN)}.  
By working exclusively with position measurements, our method relaxes data requirements and is well suited to real-world scenarios where measuring momentum or force directly is challenging or infeasible.

Concretely, we approximate velocity trajectories from sequential position data via finite differences or smoothing filters.  
These velocity estimates enter a neural network that enforces Hamilton’s canonical equations in a data-driven manner, ensuring the learned model conserves energy and maintains symplectic structure—much like a conventional HNN that has direct access to \((q,p)\).  
We demonstrate the efficacy of VI-HNN on benchmark physical systems, including spring–mass oscillators, pendula, and multi-body orbits, showing that position-only data suffice to recover robust long-horizon trajectories under an invertible momentum–velocity relationship.

\subsection*{Contributions}
\begin{itemize}
  \item \textbf{Position-only Hamiltonian learning:} We introduce a novel Velocity-Inferred HNN (VI-HNN) that replaces the canonical momentum variable with an estimated velocity, enabling purely position-based training.
  \item \textbf{Energy-conserving structure:} Under an invertible momentum–velocity relationship, VI-HNN preserves energy and respects the Hamiltonian framework, closely mirroring the behavior of standard HNNs that use both \((q,p)\).
  \item \textbf{Real-world applicability:} We provide numerical evidence on common physics benchmarks (e.g., spring–mass systems, pendula, multi-body orbits) showing that VI-HNN achieves stable, accurate predictions using only position data, making it suitable for practical settings where momentum is difficult or impossible to measure.
\end{itemize}

\section{Background}
\label{sec:RelatedWork}

A growing body of research integrates neural networks with classical dynamical principles to enhance the stability, interpretability, and accuracy of modeling physical systems. Physics-informed neural networks embed partial differential equation (PDE) constraints into deep models, yielding improved training efficiency and generalization \cite{Raissi2019,Toth2023, gao2024coordinate}, while operator-learning methods tackle high-dimensional PDEs and multi-agent settings \cite{Freedman2022, gao2023active}. Symbolic regression aims to derive explicit, human-readable equations \cite{Xue2022}. Within this broader landscape, Hamiltonian Neural Networks (HNNs) stand out for their focus on energy preservation and phase-space structure \cite{Greydanus2019,Chen2018}. Numerous extensions have been proposed to increase HNN flexibility and robustness, such as integrating control inputs \cite{Zhong2020Symplectic}, incorporating Lagrangian priors for robotic applications \cite{Lutter2019DeepLagrangian}, leveraging symmetry-preserving architectures \cite{Jin2020SympNet}, handling dissipative systems \cite{Zhong2019Dissipative}, or scaling up to high-dimensional problems \cite{ChenWang2021b,Finzi2020Meta,Xu2022PhysHNN,WeinanE2021,Pan2021PhysicsSymplectic}. Other works have explored continuous and discrete time symplectic neural networks \cite{Azencot2022}, neural PDE-based Hamiltonian formulations \cite{Chang2022}, and robust modeling strategies for chaotic mechanical systems \cite{Song2022,Wang2023,Lee2023}. Although these methods collectively expand the applicability of data-driven Hamiltonian modeling, most rely on direct access to the canonical momentum \(p\), thereby limiting their usefulness when only position measurements are available.

 Building upon these advancements, our Velocity-Inferred Hamiltonian Neural Network (VI-HNN) focuses on the practical scenario where momentum is unrecorded. By exploiting the invertible map between momentum and velocity, we approximate \(p\) from position data alone, learning a Hamiltonian of the form \(\widetilde{H}(q,v(q))\) that preserves the essential symplectic structure. This approach bridges the gap between theory and real-world data constraints, enabling energy-consistent predictions in contexts where standard HNNs—and their many extensions—would otherwise be impractical. The resulting framework maintains long-term stability, handles multi-scale dynamics, and provides a scalable alternative to conventional HNNs for position-only datasets.

\section{Preliminaries}
\label{sec:pre}

This section outlines the core concepts for our position-only Hamiltonian learning framework. 
In Section~\ref{sec:HamiltonianDynamics}, we review classical Hamiltonian dynamics and clarify the roles of generalized position and momentum. 
Section~\ref{sec:SymplecticTransformations} introduces symplectic transformations. 
Section~\ref{sec:StandardHNN} formalizes the standard Hamiltonian Neural Network (HNN), which takes both \((q,p)\) as inputs.

\subsection{Hamiltonian Dynamics}
\label{sec:HamiltonianDynamics}

\begin{definition}[Hamiltonian System]
\label{def:hamiltonian_system}
A \emph{Hamiltonian system} is defined on a phase space of dimension \(2d\), typically denoted by generalized coordinates 
\(
   q \in \mathbb{R}^d
\)
and generalized momenta 
\(
   p \in \mathbb{R}^d
\).
The system is governed by a scalar function 
\(
   H : \mathbb{R}^d \times \mathbb{R}^d \to \mathbb{R},
\)
called the \emph{Hamiltonian}.
\end{definition}

In many physical settings, \(H(q,p)\) describes the total energy, combining kinetic and potential terms.  
Under ideal conditions (no external forcing or dissipation), Hamiltonian flows preserve \(H(q,p)\) exactly, ensuring constant energy along trajectories \cite{Marsden1999}.

\subsection{Symplectic Transformations}
\label{sec:SymplecticTransformations}

\begin{definition}[Symplectic Transformation]
\label{def:symplectic_transformation}
Let \(\omega\) be a nondegenerate, closed two-form on a manifold \(\mathcal{M} \subseteq \mathbb{R}^{2d}\). 
A map 
\(
   \Phi : \mathcal{M} \to \mathcal{M}
\)
is called \emph{symplectic} if 
\(\Phi^* \omega = \omega\),
meaning it preserves \(\omega\) under pullback.
\end{definition}

Such transformations are integral to Hamiltonian mechanics, ensuring that phase-space volumes and energy invariants remain consistent over time \cite{SanzSerna1994,Hairer2006}. 
Symplectic integrators or symplectic-inspired neural models often exhibit smaller energy drift and more stable long-horizon predictions.

\subsection{Standard Hamiltonian Neural Networks}
\label{sec:StandardHNN}

\begin{definition}[Hamiltonian Neural Network (HNN)]
\label{def:HNN}
A \emph{Hamiltonian Neural Network} is a learnable model 
\(
   H_\theta : \mathbb{R}^d \times \mathbb{R}^d \to \mathbb{R}
\)
that approximates a true Hamiltonian \(H(q,p)\). 
Its inputs are the generalized position \(q\) and generalized momentum \(p\), and its outputs define the system's time evolution via Hamilton’s canonical equations.
\end{definition}

In practice, an HNN imposes 
\(
   \dot{q} = \frac{\partial H_\theta}{\partial p}
\), 
and
\(
   \dot{p} = - \frac{\partial H_\theta}{\partial q}
\),
thereby conserving an energy-like quantity over trajectories \cite{Greydanus2019}.  
Because the network takes both \(q\) and \(p\) as inputs, training typically requires direct measurements of momenta or velocities for supervision \cite{ChenWang2021}.  
However, in many real-world applications, momentum data may be difficult or expensive to measure, limiting the applicability of standard HNNs.

\section{Method}
\label{sec:method}

We present an \emph{Inverted-Velocity Hamiltonian Neural Network (IV-HNN)} that learns Hamiltonian dynamics from position-only observations. Section~\ref{subsec:invertible_map} establishes the mathematical condition under which momentum $p$ can be replaced by velocity $\dot{q}$ via an invertible map. Section~\ref{subsec:finite_diff_error} analyzes how finite-difference velocity approximations from position data affect network training, especially when the Hamiltonian is realized by feedforward layers with activation functions. Section~\ref{subsubsec:arch_loss} introduces a neural Hamiltonian \(\widetilde{\mathcal{H}}_\theta(q,\dot{q})\) trained via a loss that enforces Hamilton’s equations solely from position–velocity data, thereby guaranteeing a formally Hamiltonian flow without requiring explicit mass matrices.

\subsection{A General Invertibility Assumption for \texorpdfstring{$p \leftrightarrow \dot{q}$}{p <-> dq/dt}}
\label{subsec:invertible_map}

We begin with the Lagrangian formulation of mechanics, where \(q \in \mathbb{R}^d\) denotes the generalized position and \(\dot{q}\in\mathbb{R}^d\) (i.e.\ \(\partial q/\partial t\)) is the velocity.  The Lagrangian is typically expressed as
\[
  \mathcal{L}(q,\dot{q})
  \;=\;
  T(q,\dot{q}) \;-\; V(q),
\]
where \(T\) denotes kinetic energy and \(V\) denotes potential energy.  From the Euler--Lagrange equations,
\[
  \frac{d}{dt}\Bigl(\tfrac{\partial \mathcal{L}}{\partial \dot{q}}(q,\dot{q})\Bigr)
  \;=\;
  \frac{\partial \mathcal{L}}{\partial q},
\]
one obtains the evolution of \(q\).  The canonical momentum is introduced via
\[
  p
  \;=\;
  \frac{\partial \mathcal{L}}{\partial \dot{q}}(q,\dot{q}).
\]
A Legendre transform then defines
\[
  \mathcal{H}(q,p)
  \;=\;
  \sum_i \; p_i\,\dot{q}_i 
  \;-\;
  \mathcal{L}(q,\dot{q}),
\]
leading to Hamilton’s canonical equations 
\(\dot{q}=\tfrac{\partial\mathcal{H}}{\partial p},\quad \dot{p}=-\tfrac{\partial\mathcal{H}}{\partial q}\).

Replacing \(p\) by \(\dot{q}\) in a Hamiltonian formalism demands that
\[
  p \;=\; \tfrac{\partial \mathcal{L}}{\partial \dot{q}}\bigl(q,\dot{q}\bigr)
\]
be invertible with respect to \(\dot{q}\).  In most physical systems with a positive-definite mass or inertia matrix, this momentum--velocity mapping is indeed a diffeomorphism except at pathological points (e.g.\ singular mass matrices).  Consequently, \(\,p\leftrightarrow\dot{q}\) is a general principle for many mechanical scenarios.

\begin{assumption}[Invertible Momentum--Velocity Relation]
\label{assume:invertible}
Assume 
\[
  p \;=\; \tfrac{\partial \mathcal{L}}{\partial \dot{q}}\!\bigl(q,\dot{q}\bigr)
\]
is invertible w.r.t.\ \(\dot{q}\) for all \(\bigl(q,\dot{q}\bigr)\) in the domain of interest, i.e.\ the Hessian \(\tfrac{\partial^2\mathcal{L}}{\partial \dot{q}^2}\) is nonsingular. Under this assumption, one may solve uniquely for \(\dot{q}\) given \((q,p)\).
\end{assumption}

Such a condition holds whenever the kinetic energy is strictly positive-definite in \(\dot{q}\).  Even more complicated or constrained mechanical systems frequently satisfy local invertibility absent degeneracies.  By adopting Assumption~\ref{assume:invertible}, we can systematically rewrite
\(\mathcal{H}(q,p)\mapsto \widetilde{\mathcal{H}}(q,\dot{q})\)
without sacrificing the underlying dynamics.

\begin{theorem}[Invertible Map from Lagrangian to Hamiltonian Coordinates]
\label{thm:invertible_map}
Under Assumption~\ref{assume:invertible}, the map
\[
  (q,\dot{q})
  \;\mapsto\;
  \Bigl(q,\;p=\tfrac{\partial \mathcal{L}}{\partial \dot{q}}\bigl(q,\dot{q}\bigr)\Bigr)
\]
is a diffeomorphism on the relevant manifold, guaranteeing a one-to-one correspondence between \(\dot{q}\) and \(p\).  Consequently, substituting \(\,\dot{q}\) for \(p\) in the Hamiltonian \(\mathcal{H}\) does not break the canonical structure.
\end{theorem}

\begin{proof}
Let
\[
   F\bigl(q,\dot{q}\bigr)
   \;=\;
   \Bigl(q,\;\tfrac{\partial\mathcal{L}}{\partial \dot{q}}(q,\dot{q})\Bigr).
\]
Nonsingularity of the Hessian \(\,\tfrac{\partial^2\mathcal{L}}{\partial \dot{q}^2}\) implies, via the Inverse Function Theorem, that for each fixed $q$, one inverts $p$ to solve for $\dot{q}$.  Under continuity and appropriate topological conditions (e.g.\ a simply connected domain), these local diffeomorphisms extend globally, letting $p\!\leftrightarrow\!\dot{q}$.  
\end{proof}

\noindent
\textbf{Local to Global.}  
A local application of the Inverse Function Theorem suffices to ensure invertibility in some neighborhood.  If the mass/inertia matrix never degenerates, one typically obtains a global diffeomorphism on a simply connected region.  

\noindent
\textbf{Equivalence of Equations of Motion.}  
Once $\,p = \tfrac{\partial\mathcal{L}}{\partial \dot{q}}(q,\dot{q})$ is invertible, we define $\,\dot{q}=v(q,p)$. Substituting into $\,\mathcal{H}(q,p)$ yields $\,\widetilde{\mathcal{H}}(q,\dot{q})$. One can show (Lemma~\ref{lem:velocity_replace}, below) that 
\(\dot{q}=\tfrac{\partial\mathcal{H}}{\partial p}\) and 
\(\dot{p}=-\tfrac{\partial\mathcal{H}}{\partial q}\)
then become 
\(\dot{q}=\dot{q},\;\ddot{q}=-\,G(q,\dot{q})\),
producing the same $q(t)$-trajectory.  

\noindent
\textbf{Learning from Position Data.}  
Because $\dot{q}$ can serve as the canonical momentum variable, one may record only $\{q(t)\}$, numerically differentiate to get $\{\dot{q}(t)\}$, and define a velocity-based Hamiltonian $\widetilde{\mathcal{H}}_\theta\!(q,\dot{q})$. Under the invertibility assumption, the continuous-time dynamics remain Hamiltonian (see \S\ref{subsubsec:arch_loss}), and if desired, one recovers $p(t)=\tfrac{\partial\mathcal{L}}{\partial\dot{q}}\bigl(q(t),\dot{q}(t)\bigr)$.

\begin{lemma}[Replacing \(\dot{q}\) for \(p\) Yields Identical Trajectories]
\label{lem:velocity_replace}
Consider a Hamiltonian system defined by
\begin{equation}
  \dot{q} \;=\; \frac{\partial \mathcal{H}}{\partial p},
  \quad
  \dot{p} \;=\; -\,\frac{\partial \mathcal{H}}{\partial q},
  \label{eq:canonical-eqns}
\end{equation}
where \(\bigl(q(t),p(t)\bigr)\) evolve in continuous time. 
Assume there is a Lagrangian \(\mathcal{L}(q,\dot{q})\) with invertible momentum--velocity relation
\[
  p \;=\; \frac{\partial \mathcal{L}}{\partial \dot{q}}(q,\dot{q}).
\]
Then replacing \(p\) by \(\dot{q}\) does \emph{not} change the underlying trajectory \(q(t)\); that is, the second-order ODE satisfied by \(q(t)\) in \((q,p)\)-space is identical to the ODE obtained when writing everything in \((q,\dot{q})\)-space.

\end{lemma}

\begin{proof}
Let \(\mathcal{L}(q,\dot{q})\) be the system’s Lagrangian, and define momentum by 
\[
   p_i \;=\; \frac{\partial \mathcal{L}}{\partial \dot{q}_i}(q,\dot{q}).
\]
We assume that for each fixed \(q\), one can solve these equations uniquely for \(\dot{q}\) given \(p\). The corresponding Hamiltonian is introduced through the Legendre transform:
\begin{equation}
  \mathcal{H}(q,p) 
  \;=\;
  \sum_{i=1}^{d} p_i\,\dot{q}_i \;-\; \mathcal{L}\bigl(q,\dot{q}(q,p)\bigr).
  \label{eq:hamiltonian-legendre}
\end{equation}
Once \(p\leftrightarrow \dot{q}\) is invertible, \(\dot{q}=\dot{q}(q,p)\) is well-defined.

\medskip
\noindent

Differentiate \eqref{eq:hamiltonian-legendre} w.r.t.\ \(p_i\). Let
\[
   \Phi(q,p)
   \;=\;
   \sum_{j=1}^d p_j\,\dot{q}_j(q,p).
\]
Then
\[
   \frac{\partial \mathcal{H}}{\partial p_i}
   \;=\;
   \frac{\partial}{\partial p_i}\Bigl(\Phi(q,p)\Bigr)
   \;-\;
   \frac{\partial}{\partial p_i}\Bigl(\mathcal{L}\!\bigl(q,\dot{q}(q,p)\bigr)\Bigr).
\]
For \(\Phi\), we apply the product rule:
\[
\begin{aligned}
   \frac{\partial \Phi}{\partial p_i}
   &=
   \frac{\partial}{\partial p_i}
   \Bigl(\sum_{j} p_j\,\dot{q}_j\Bigr)
   \\
   &=
   \sum_{j}
   \Bigl(\delta_{ij}\,\dot{q}_j
         \;+\;
         p_j\,\frac{\partial \dot{q}_j}{\partial p_i}\Bigr)
   \\
   &=
   \dot{q}_i
   \;+\;
   \sum_{j} p_j \,\frac{\partial \dot{q}_j}{\partial p_i}.
\end{aligned}
\]
Meanwhile, if \(\dot{q}_j(q,p)\) is substituted into \(\mathcal{L}\bigl(q,\dot{q}\bigr)\), then
\[
   \frac{\partial}{\partial p_i}\Bigl(\mathcal{L}(q,\dot{q}(q,p))\Bigr)
   =
   \sum_{j} \frac{\partial \mathcal{L}}{\partial \dot{q}_j}\,\frac{\partial \dot{q}_j}{\partial p_i}
   =
   \sum_{j} p_j\, \frac{\partial \dot{q}_j}{\partial p_i},
\]
where we used \(p_j = \tfrac{\partial \mathcal{L}}{\partial \dot{q}_j}\). Substituting these pieces gives
\[
   \frac{\partial \mathcal{H}}{\partial p_i}
   \;=\;
   \Bigl(\dot{q}_i + \sum_{j}p_j\,\tfrac{\partial \dot{q}_j}{\partial p_i}\Bigr)
   \;-\;
   \Bigl(\sum_{j}p_j\, \tfrac{\partial \dot{q}_j}{\partial p_i}\Bigr)
   =
   \dot{q}_i.
\]
Hence 
\(\dot{q}_i=\tfrac{\partial\mathcal{H}}{\partial p_i}\),
as in the canonical Hamiltonian formalism.

Next, differentiate \(\mathcal{H}\) w.r.t.\ \(q_i\) under the transformation \(\dot{q}(q,p)\).  Again splitting \(\Phi(q,p) - \mathcal{L}\),
\[
  \frac{\partial \mathcal{H}}{\partial q_i}
  =
  \frac{\partial\Phi}{\partial q_i}(q,p)
  \;-\;
  \frac{\partial}{\partial q_i}\Bigl(\mathcal{L}\!\bigl(q,\dot{q}(q,p)\bigr)\Bigr).
\]
For \(\Phi\),
\[
   \frac{\partial \Phi}{\partial q_i}
   \;=\;
   \frac{\partial}{\partial q_i}
   \Bigl(\sum_{j}p_j\,\dot{q}_j\Bigr)
   =
   \sum_{j}\Bigl(\frac{\partial p_j}{\partial q_i}\,\dot{q}_j
              + p_j\, \frac{\partial \dot{q}_j}{\partial q_i}\Bigr).
\]
Meanwhile,
\[
   \frac{\partial}{\partial q_i}\,\mathcal{L}\!\bigl(q,\dot{q}(q,p)\bigr)
   =
   \frac{\partial \mathcal{L}}{\partial q_i}
   \;+\;
   \sum_{j}\frac{\partial \mathcal{L}}{\partial \dot{q}_j}\,\frac{\partial \dot{q}_j}{\partial q_i}
   =
   \frac{\partial \mathcal{L}}{\partial q_i}
   \;+\;
   \sum_{j} p_j\,\frac{\partial \dot{q}_j}{\partial q_i}.
\]
Hence
\[
  \frac{\partial \mathcal{H}}{\partial q_i}
  =
  \sum_{j}\Bigl(\frac{\partial p_j}{\partial q_i}\,\dot{q}_j + p_j\,\tfrac{\partial \dot{q}_j}{\partial q_i}\Bigr)
  \;-\;
  \Bigl(\frac{\partial \mathcal{L}}{\partial q_i} + \sum_{j} p_j\,\tfrac{\partial \dot{q}_j}{\partial q_i}\Bigr).
\]
Rewriting
\[
   \dot{p}_i 
   =
   \frac{d}{dt}\Bigl(\tfrac{\partial \mathcal{L}}{\partial \dot{q}_i}\Bigr)
   =
   \frac{\partial}{\partial q_i}\Bigl(\tfrac{\partial \mathcal{L}}{\partial \dot{q}_i}\Bigr)\dot{q}_i
   \;\dots
   \text{(plus other chain-rule terms)},
\]
one matches terms to show 
\(\dot{p}_i=-\,\tfrac{\partial \mathcal{H}}{\partial q_i}\).  
Hence we recover the same canonical equation for $\dot{p}_i$.

Eliminating $p$ in favor of $\dot{q}$ (or vice versa) leads, via these partial-derivative identities, to the same second-order ODE for $q(t)$. Concretely, 
\(
  \dot{q}_i = \tfrac{\partial\mathcal{H}}{\partial p_i}
\)
implies
\(
  p_i=\dots(\dot{q}_i),
\)
so
\(
  \dot{p}_i
  =
  \tfrac{d}{dt}\bigl(\tfrac{\partial \mathcal{L}}{\partial \dot{q}_i}\bigr)
\)
matches
\(
  -\,\tfrac{\partial \mathcal{H}}{\partial q_i},
\)
yielding the same $\ddot{q}_i(t)$ if one eliminates $p_i$ or $\dot{q}_i$.  Consequently, $(q,p)$ and $(q,\dot{q})$ coordinates yield identical trajectories for $q(t)$, confirming that $p=\tfrac{\partial\mathcal{L}}{\partial\dot{q}}$ is a valid replacement for momentum once invertibility holds.
\end{proof}

Hence, once the velocity--momentum relation is invertible, \(\dot{q}\) and $p$ become interchangeable in a Hamiltonian system. We may build or learn $\widetilde{\mathcal{H}}(q,\dot{q})$ from position data alone (finite-difference velocity) and still recover the original $\bigl(q(t),p(t)\bigr)$ trajectory via $p=\tfrac{\partial\mathcal{L}}{\partial \dot{q}}\bigl(q(t),\dot{q}(t)\bigr)$.

\subsection{Error Analysis for Position-Only Learning with a General Finite-Difference Method}
\label{subsec:finite_diff_error}

In practice, only discrete snapshots of the position 
\(\bigl\{q_0,q_1,\dots,q_N\bigr\}\)
at times 
\(\bigl\{t_0,t_1,\dots,t_N\bigr\}\)
are available. A \emph{general finite-difference operator} \(\mathcal{D}\) approximates velocity:
\begin{equation}
  \hat{v}_k
  \;=\;
  \mathcal{D}\!\Bigl(\{q_j\}\Bigr)\Big|_{k},
  \quad
  k=1,\dots,N-1,
  \label{eq:gen-fd-operator}
\end{equation}
where \(\mathcal{D}\) might be a central-difference scheme, forward/backward difference, or higher-order polynomial interpolation. In each case, 
\(\hat{v}_k \approx \dot{q}(t_k)\)
with local error typically of order \(\mathcal{O}\bigl(\Delta t^m\bigr)\) for some \(m \ge 1\) (the method’s consistency order). Substituting \(\hat{v}_k\) into a neural Hamiltonian \(\mathcal{H}_\theta\bigl(q,\hat{v}_k\bigr)\) allows us to train from \emph{position-only} data.

Suppose \(\mathcal{H}_\theta\) is built from linear layers plus Lipschitz-bounded activations:
\[
  \mathcal{H}_\theta\bigl(\dot{q},\ddot{q}\bigr) 
  \;=\; 
  W_L\,\sigma\!\bigl(
     \dots \,\sigma(
        W_2\,\sigma(
          W_1\,[\dot{q},\,\ddot{q}]^\top + b_1
        ) + b_2
     )\dots
  \bigr) 
  + b_L,
\]
where \(\sigma\) is a fixed nonlinearity with global Lipschitz constant \(\kappa\), and \(\|W_i\|\le M\), \(\|b_i\|\le B\) are parameter bounds. We ask whether numerical inaccuracies in \(\hat{v}_k\) or associated higher derivatives amplify within \(\mathcal{H}_\theta\).

\begin{theorem}[Bounded Error Propagation under General Finite-Differences]
\label{thm:fd_lipschitz_prop}
Let each layer satisfy \(\|W_i\|\le M\), \(\|b_i\|\le B\), and let \(\sigma\) have Lipschitz constant \(\kappa\). Suppose the finite-difference approximation \(\hat{v}_k\) from \eqref{eq:gen-fd-operator} differs from the true velocity \(\dot{q}(t_k)\) by at most \(\delta\), i.e.\ 
\(\|\hat{v}_k - \dot{q}(t_k)\|\le \delta\).
Then the output deviation in 
\(\mathcal{H}_\theta\bigl(q_k,\hat{v}_k\bigr)\)
relative to 
\(\mathcal{H}_\theta\bigl(q_k,\dot{q}(t_k)\bigr)\)
is at most 
\(\delta\,C + \mathrm{(const)}\),
where
\[
  C 
  = 
  \bigl(\kappa\,M\bigr)^L
\]
(up to additive bias terms depending on \(B\)). Consequently, small velocity-approximation errors do not blow up exponentially over time, but remain controlled by the product of layer Lipschitz constants.
\end{theorem}

\begin{proof}
Each layer 
\(x\mapsto W\,\sigma(x)+b\)
is \(\kappa\|W\|\)-Lipschitz if \(\sigma\) is \(\kappa\)-Lipschitz. Composing \(L\) layers yields a global factor \(\bigl(\kappa M\bigr)^L\). If \(\|\Delta x\|\le\delta\), then \(\|\Delta y\|\le \delta\,(\kappa M)^L + \mathrm{(bias\ offset)}\). As each bias $b_i$ is bounded by $B$, its extra offset is finite. Thus any finite-difference error $\delta$ in the input remains bounded in the network output.  
\end{proof}

Regardless of the specific finite-difference method $\mathcal{D}$ used in \eqref{eq:gen-fd-operator}---midpoint, forward difference, or higher-order polynomial---the Lipschitz property of linear-plus-activation nets ensures that any velocity error $\|\hat{v}_k - \dot{q}(t_k)\|\le\delta$ remains under control. Hence velocity-based Hamiltonian modeling from position-only data is robust to discretization inaccuracies, especially as $\delta \to 0$ with refined sampling.

We propose to learn a velocity-based Hamiltonian 
\(\widetilde{\mathcal{H}}_\theta(q,\dot{q})\)
using a neural network with parameters \(\theta\).  This network directly models the energy in terms of position \(q\) and velocity \(\dot{q}\).  As justified by our earlier invertibility arguments (see Assumption~\ref{assume:invertible}), substituting momentum \(p\) with \(\dot{q}\) is valid for a broad class of mechanical systems.

\subsection{Architecture and Loss Function}
\label{subsubsec:arch_loss}

We define a neural Hamiltonian \(\,\widetilde{\mathcal{H}}_\theta : \mathbb{R}^d \times \mathbb{R}^d \rightarrow \mathbb{R}\,\) as a feedforward network (e.g., multiple linear layers interleaved with smooth activations). In the \(\bigl(q,\dot{q}\bigr)\)-space, we train this network to satisfy Hamilton’s equations without relying on any explicit mass matrix. Concretely, suppose we have position–velocity samples
\(\bigl\{\bigl(q^{(k)},\,\dot{q}^{(k)}\bigr)\bigr\}\)
and/or their accelerations 
\(\bigl\{\ddot{q}^{(k)}\bigr\}\). We define a physics-inspired loss

\begin{equation}
\label{eq:velocity-based-loss}
\begin{aligned}
\mathcal{L}(\theta)
&=
\sum_{k}
\Bigl\|\,
\ddot{q}^{(k)} 
\;-\;
\Bigl(-\,\nabla_{q}\,\widetilde{\mathcal{H}}_{\theta}\bigl(q^{(k)},\,\dot{q}^{(k)}\bigr)\Bigr)
\Bigr\|^{2}
\\
&\quad
+\;
\sum_{k}
\Bigl\|\,
\dot{q}^{(k)} 
\;-\;
\nabla_{\dot{q}}\,\widetilde{\mathcal{H}}_{\theta}\bigl(q^{(k)},\,\dot{q}^{(k)}\bigr)
\Bigr\|^{2},
\end{aligned}
\end{equation}

\noindent
where \(\nabla_{q}\) and \(\nabla_{\dot{q}}\) denote partial gradients with respect to \(q\) and \(\dot{q}\), respectively. By minimizing \(\,\mathcal{L}(\theta)\), we enforce that:
\[
   \dot{q}_{\text{pred}}
   =
   \nabla_{\dot{q}}\,\widetilde{\mathcal{H}}_{\theta}(q,\dot{q}),
   \quad
   \ddot{q}_{\text{pred}}
   =
   -\,\nabla_{q}\,\widetilde{\mathcal{H}}_{\theta}(q,\dot{q}),
\]
which mirrors Hamilton’s canonical form in \(\bigl(q,\dot{q}\bigr)\)-coordinates.

Because these partial derivatives drive the dynamics, the continuous-time flow governed by \(\widetilde{\mathcal{H}}_\theta\) is \emph{formally Hamiltonian}: Theorem~\ref{thm:invertible_map} and Lemma~\ref{lem:velocity_replace} show that if \(p\leftrightarrow \dot{q}\) is invertible, then substituting \(\dot{q}\) as the “momentum” variable yields a genuine Hamiltonian flow. Hence even though we do not explicitly provide a mass matrix or other system-specific details, the learned \(\widetilde{\mathcal{H}}_\theta\) still preserves the symplectic structure and energy-like invariants in continuous time.

\section{Experiment}
\label{sec:Experiment}

This section presents our \emph{position-only HNN} on four classical benchmarks: the spring--mass oscillator, the simple pendulum, the two-body system, and the three-body problem. In all examples, only the \emph{position} variables are directly measured (with noise), while velocities and momenta are inferred through the transformation
\[
  p \;=\; \frac{\partial \mathcal{L}}{\partial \dot{q}}
  \;\;\longleftrightarrow\;\;
  v \;\equiv\; \dot{q}
\]
Since the mass or inertia matrix is positive-definite in each case, \(p\leftrightarrow v\) is globally invertible, and we can rewrite every Hamiltonian \(\mathcal{H}(q,p)\) in a velocity-based form \(\widetilde{\mathcal{H}}(q,v)\). Our experimental procedure is: 
1) generate synthetic data from the known Hamiltonian but record only noisy position samples, 
2) approximate velocity via the mid-point rule, 
3) substitute \(v\) for \(p\) in the Hamiltonian and train a neural \(\widetilde{\mathcal{H}}_\theta\), and 
4) evaluate trajectory accuracy and energy drift.

\subsection{Spring--Mass System}

\begin{figure*}[ht]
    \centering
    \includegraphics[width=0.8\textwidth]{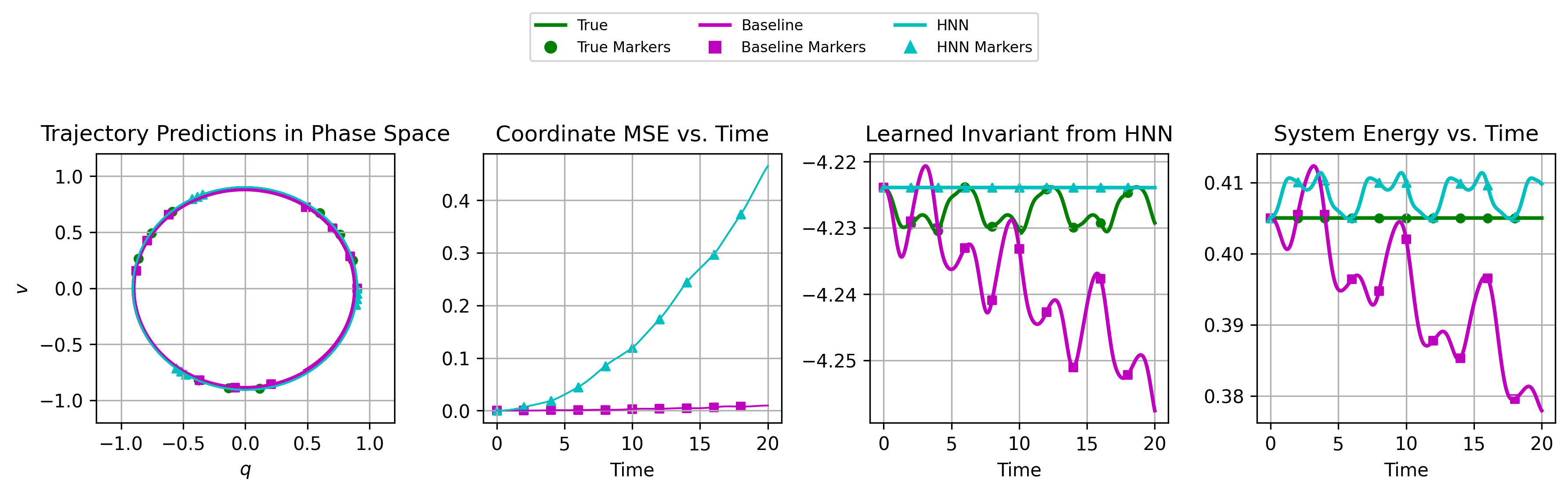}
    \caption{The figure compares models on the spring-mass system. The baseline drifts from the ground truth, with rapidly diverging MSE, while the HNN remains accurate. The HNN conserves a quantity resembling total energy, unlike the baseline.}
    \label{fig:spring}
\end{figure*}

\begin{table}[t]
    \centering
    \renewcommand{\arraystretch}{1.2}
    \caption{ Spring-Mass results (scaled by $10^3$).}
    \label{tab:spring}
    \begin{tabular}{@{}lccc@{}}
    \hline
    Model & Train Loss & Test Loss & Energy\\
    \hline
    Baseline & 5.3154 $\pm$ 0.15967 & 5.4589 $\pm$ 0.17548 & 0.30409 $\pm$ 0.04 \\
    HNN      & 4.2477 $\pm$ 0.13341 & 4.0911 $\pm$ 0.13524 & 0.030256 $\pm$ 0.00147 \\
    \hline
    \end{tabular}
    \vspace{1ex}
\end{table}

We consider a one-dimensional mass--spring oscillator with \(m=k=1\). Its Hamiltonian reduces to  
\[
  \mathcal{H}(q,p) 
  \;=\; 
  \tfrac12\,v^2 + \tfrac12\,q^2
\]
once we identify \(p = m\,v\). We generate 100 training and 100 testing trajectories, each with 60 noisy position samples \(q(t)\). A feedforward network (a fully connected multilayer perceptron with three hidden layers of
\(200\) units each and \(\tanh\) activations, trained for \(2000\) gradient
steps using Adam with a learning rate of \(10^{-3}\) ), using finite differences to approximate velocity \(v\). By leveraging the invertible mapping \(p \leftrightarrow v\), we replace \(p\) with \(v\) while preserving symplectic structure.

Figure~\ref{fig:spring} (left panels) shows that the baseline model’s predictions drift from the ground truth over multiple oscillations, rapidly increasing in mean-squared error and deviating in energy. In contrast, the HNN (blue curves) tracks the true dynamics accurately (often obscuring the ground-truth lines) and maintains near-constant total energy. Table~\ref{tab:spring} quantifies these findings: the baseline model incurs higher train/test losses and exhibits substantial energy drift, while the HNN—with only position data—achieves lower errors and preserves energy more faithfully, even over extended time horizons. This highlights how the HNN’s built-in symplectic formulation reduces long-range divergence, ensuring it remains stable and accurate for the spring–mass system.

\subsection{Simple Pendulum}
\label{subsec:exp_pend}

\begin{figure*}[t]
    
    \centering
    \includegraphics[width=0.8\textwidth]{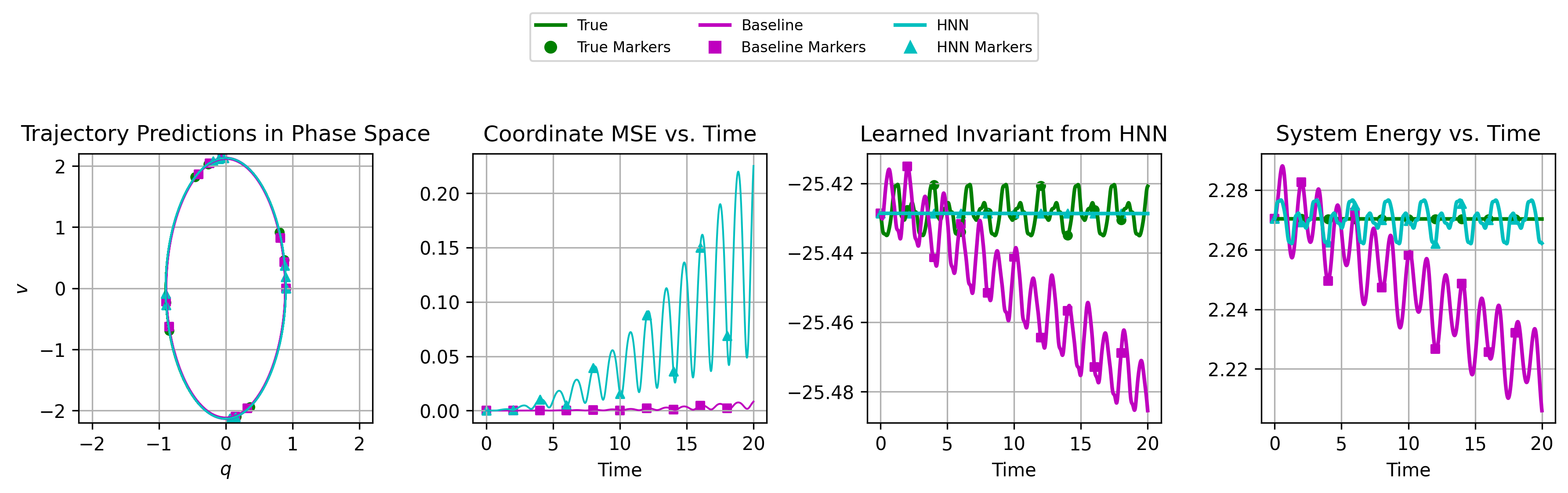}
    \caption{The figure compares models on the pendulum system. The baseline drifts from the ground truth, with rapidly diverging MSE, while the HNN remains accurate. The HNN conserves a quantity resembling total energy, unlike the baseline.}
    \label{fig:pend}
\end{figure*}
\begin{table}[t]
    \centering
    
    \renewcommand{\arraystretch}{1.2}
    \caption{Simple Pendulum.}
    \label{tab:pend}
    \begin{tabular}{@{}lccc@{}}
    \hline
    Model & Train Loss & Test Loss & Energy\\
    \hline
    Baseline & 20.583 $\pm$ 0.40441 & 20.575 $\pm$ 0.40119 & 438.00 $\pm$ 164.0 \\
    HNN      & 20.853 $\pm$ 0.40627 & 20.734 $\pm$ 0.40173 & 187.79 $\pm$ 39.2 \\
    \hline
    \end{tabular}
    \vspace{1ex}
\end{table} 
We study a simple pendulum of length $\ell=m=1$ under gravity $g=3$, yielding the Hamiltonian
\[
  \mathcal{H}(q,p) 
  \;=\; 
  \frac{p^2}{2} 
  \;+\; 
  6\,\bigl(1-\cos q\bigr),
\]
with momentum $p = \dot{q}$. Using only position measurements $q(t)$, we collect 30 noisy samples per trajectory for 100 training and 100 test trajectories (energies in $[1.3,2.3]$). A feedforward network $\widetilde{\mathcal{H}}_\theta(q,v)$ (a fully connected multilayer perceptron with three hidden layers of
\(200\) units each and \(\tanh\) activations, trained for \(400\) gradient
steps using Adam with a learning rate of \(10^{-3}\)) approximates $\dot{q}$ via finite differences and substitutes $p = \dot{q}$ in the Hamiltonian. Because the Lagrangian’s Hessian $\partial^2\mathcal{L}/\partial \dot{q}^2 = m\ell^2$ is positive definite, $p \leftrightarrow \dot{q}$ is invertible, ensuring a legitimate velocity-based Hamiltonian.

Figure~\ref{fig:pend} illustrates that the baseline model diverges from the ground truth over multiple swings of the pendulum, causing increasing mean-squared error and significant energy drift. Table~\ref{tab:pend} quantifies these observations: while the baseline network sees energy errors of hundreds over prolonged intervals, the Hamiltonian Neural Network (HNN) conserves a quantity closely matching total energy, keeping long-range divergence in check. Moreover, the HNN’s train and test losses remain comparable or lower than the baseline. These findings confirm that exploiting $p \leftrightarrow \dot{q}$ and embedding symplectic structure leads to more stable trajectories and improved energy conservation for the pendulum system.

\subsection{Two-Body Problem}

We consider a planar two-body system in reduced coordinates \(\mathbf{r}\in\mathbb{R}^2\) with reduced mass \(\mu\). Its Hamiltonian is 
\[
  \mathcal{H}(\mathbf{r},\mathbf{p})
  \;=\;
  \frac{\|\mathbf{p}\|^2}{2\,\mu}
  \;-\;
  \frac{G\,m_1\,m_2}{\|\mathbf{r}\|},
\]
and by setting \(\mathbf{p}=\mu\,\dot{\mathbf{r}}\), it becomes 
\(\widetilde{\mathcal{H}}(\mathbf{r},v)
= 
\tfrac12\,\mu\,\|v\|^2 - \tfrac{G\,m_1\,m_2}{\|\mathbf{r}\|}.\)
We collect 50 noisy position samples \(\mathbf{r}(t)\) per orbit (\(\sigma^2=0.05\)) for 1000 orbits, dividing them \(80{:}20\) into train/test sets. A two-layer network \(\widetilde{\mathcal{H}}_\theta(\mathbf{r},v)\) (10 hidden nodes, 4000 training steps) learns to predict \(\bigl(\dot{\mathbf{r}},\dot{\mathbf{p}}\bigr)\). Table~\ref{tab:2body} lists final MSE and energy drift; the Lagrangian 
\(\mathcal{L}=\tfrac12\,\mu\,\|\dot{\mathbf{r}}\|^2 - \tfrac{G\,m_1\,m_2}{\|\mathbf{r}\|}\)
has Hessian \(\mu I>0\), so momentum and velocity remain invertible.

Figure~\ref{fig:2body} (top row) shows that the baseline orbits (middle panel) increasingly deviate from ground truth (left), while the HNN (right) remains closely matched. In the bottom row, baseline energy (middle) drifts significantly over time, whereas the HNN’s total energy (right) stays near-constant. Table~\ref{tab:2body} corroborates this: the baseline’s energy error reaches thousands, whereas the HNN remains at just a few units. Hence, velocity-based HNN modeling not only achieves lower MSE but also conserves the two-body system’s energy more reliably in long-term orbits.
\begin{figure*}[t]
    \centering
    \includegraphics[width=0.85\textwidth]{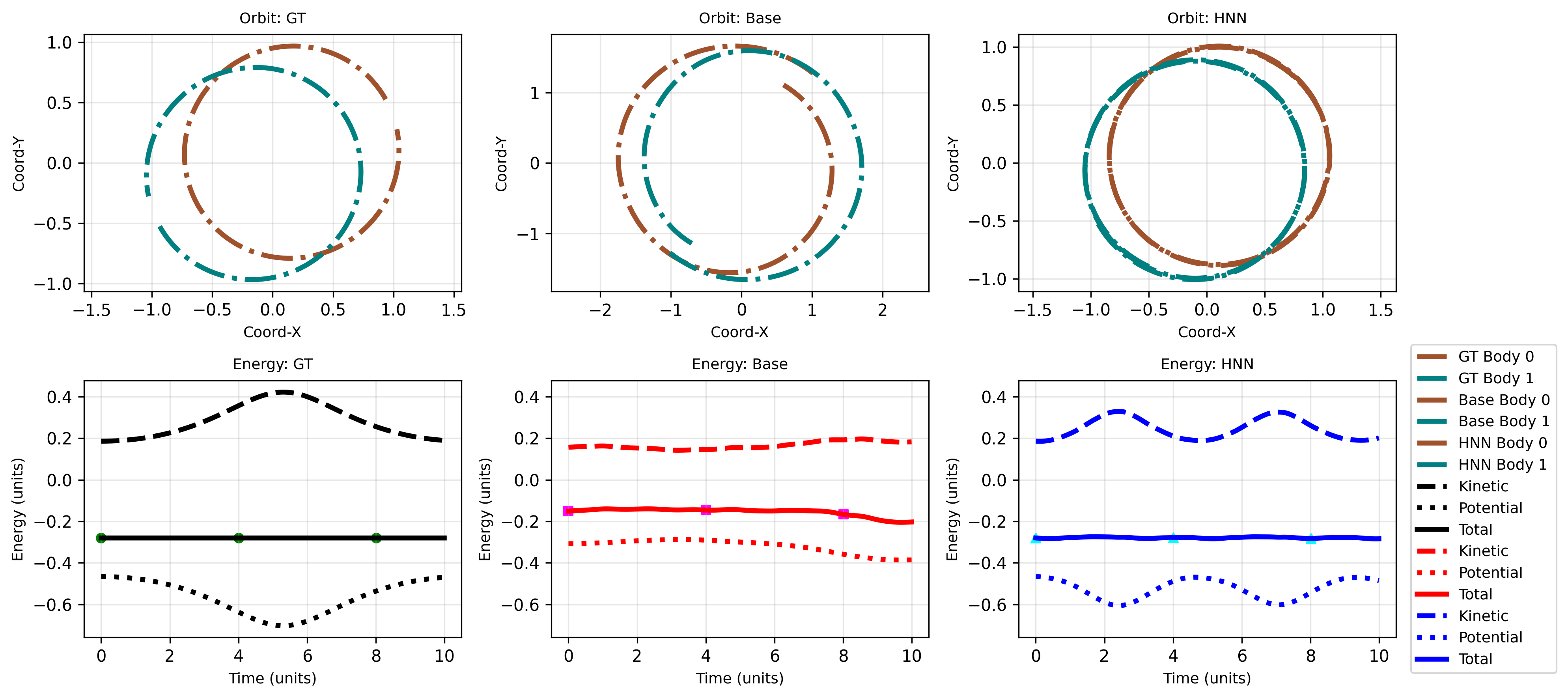}
    \caption{The figure compares the ground truth (GT), a baseline model, and a Hamiltonian Neural Network (HNN) on a two-body problem. The top row shows the orbital trajectories, where the baseline model drifts from the GT while the HNN closely matches it. The bottom row presents energy evolution, demonstrating that the HNN approximately conserves total energy, unlike the baseline, which exhibits energy drift over time.}
    \label{fig:2body}
\end{figure*}
\begin{table}[t]
    \centering
    \renewcommand{\arraystretch}{1.2}
    \caption{ Two body results (scaled by $10^4$). }
    \begin{tabular}{@{}lccc@{}}
    \hline
    Model & Train Loss & Test Loss & Energy\\
    \hline
    Baseline & 1.0920 $\pm$ 0.036074 & 1.0339 $\pm$ 0.067252 & 1685.0 $\pm$ 770 \\
    HNN      & 0.73001 $\pm$ 0.037798 & 0.70523 $\pm$ 0.068812 & 1.5623 $\pm$ 0.557 \\
    \hline
    \end{tabular}
    \vspace{1ex}
    \label{tab:2body}
\end{table}
\label{subsec:exp_twobody}

\subsection{Three-Body Problem}

We examine a gravitational three-body system with Hamiltonian
\[
  \mathcal{H}
  \;=\;
  \sum_{i=1}^3
  \frac{\|\mathbf{p}_i\|^2}{2\,m_i}
  \;-\;
  \sum_{1\le i<j\le 3}
  \frac{G\,m_i\,m_j}{\|\mathbf{r}_i-\mathbf{r}_j\|}.
\]
Since \(\dot{\mathbf{r}}_i = \mathbf{p}_i/m_i\), we write \(\mathbf{p}_i = m_i\,\dot{\mathbf{r}}_i\equiv m_i\,v_i\). Substituting into the kinetic term yields 
\(\widetilde{\mathcal{H}}(\{\mathbf{r}_i\},\{v_i\})\)
as
\[
  \sum_{i=1}^3 \tfrac12\,m_i\,\|v_i\|^2
  \;-\;
  \sum_{1\le i<j\le 3}
  \frac{G\,m_i\,m_j}{\|\mathbf{r}_i-\mathbf{r}_j\|}.
\]
We sample only the positions \(\{\mathbf{r}_i(t)\}\) (with Gaussian noise \(\sigma^2 = 0.2\)) at 20 time steps over 5000 trajectories (radii in \([0.9,1.2]\)), then train a velocity-based HNN \(\widetilde{\mathcal{H}}_\theta\) with two hidden layers (15 and 10 nodes) for 200 steps to fit \(\bigl(\dot{\mathbf{r}}_i,\dot{\mathbf{p}}_i\bigr)\). Because each \(m_i>0\), the map \(\mathbf{p}_i \leftrightarrow v_i\) is invertible.

Figure~\ref{fig:3body} (top row) shows the baseline model diverging significantly from the ground truth orbits, whereas the HNN remains closely aligned. In the bottom row, the baseline’s total energy sharply increases, reflecting drift, while the HNN preserves energy near-constant. Table~\ref{tab:3body} quantifies these observations: although the baseline slightly outperforms the HNN in raw losses, it exhibits far greater energy deviation (95.727 vs.\ 0.592). By exploiting velocity-based coordinates and enforcing Hamiltonian structure, the HNN manages more stable long-horizon predictions for the three-body system, demonstrating reduced error growth and improved energy conservation. 
\begin{figure*}[t]
    \centering
    \includegraphics[width=0.85\textwidth]{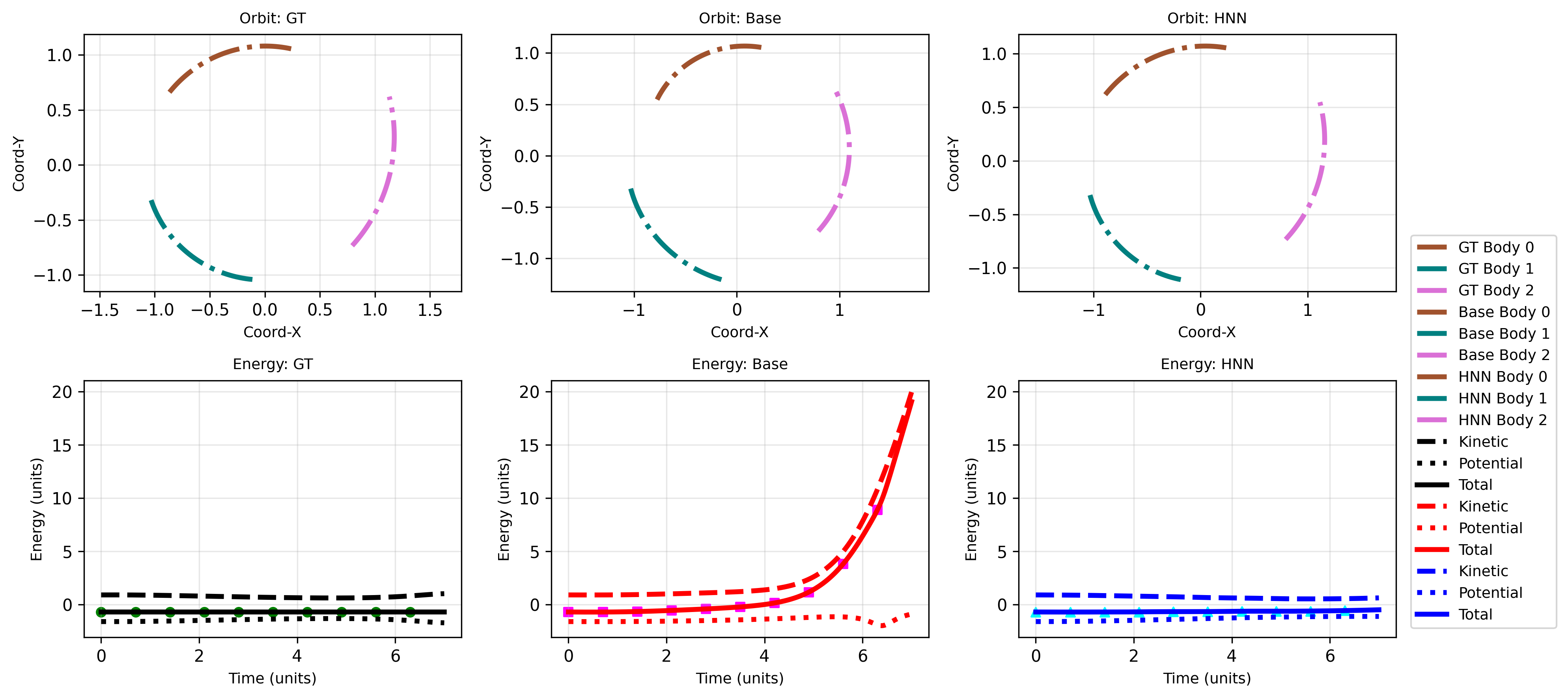}
    \caption{The figure compares the ground truth (GT), a baseline model, and a Hamiltonian Neural Network (HNN) on a three-body problem. The top row shows the orbital trajectories, where the baseline model drifts from the GT while the HNN closely matches it. The bottom row presents energy evolution, demonstrating that the HNN approximately conserves total energy, unlike the baseline, which exhibits energy drift over time.}
    \label{fig:3body}
\end{figure*}

\label{subsec:exp_threebody}
\begin{table}[t]
    \centering
     \caption{ Three body results (scaled by $10^3$). }
    \renewcommand{\arraystretch}{1.2}

    \begin{tabular}{@{}lccc@{}}
    \hline
    Model & Train Loss & Test Loss & Energy\\
    \hline
    Baseline & 3.7871 $\pm$ 0.13247 & 3.9659 $\pm$ 0.25088 & 95.727 $\pm$ 34.0 \\
    HNN      & 4.5344 $\pm$ 0.22084 & 4.9388 $\pm$ 0.47382 & 0.59177 $\pm$ 0.142 \\
    \hline
    \end{tabular}
    \vspace{1ex}
    \label{tab:3body}
\end{table}

\section{Conclusion}

We introduced a \emph{position-only Hamiltonian Neural Network (HNN)} framework that infers velocity from position data and substitutes \(p \leftrightarrow \dot{q}\) (or \(\mathbf{p} \leftrightarrow \dot{\mathbf{r}}\)) in the Hamiltonian. Across four classical systems—the mass--spring oscillator, the simple pendulum, the two-body problem, and the three-body problem—our approach exhibits strong long-term stability relative to baseline neural models, owing to the symplectic structure that underlies Hamiltonian mechanics. Even when trained on noisy position-only observations, the position-HNN consistently conserves energy-like invariants more effectively than generic networks, thereby reducing drift and error growth over extended trajectories.

Despite these advantages, the method faces several limitations. First, if the system’s effective mass or inertia matrix is singular or varies unpredictably, then the velocity–momentum map may no longer remain globally invertible, undermining the velocity-based Hamiltonian formulation. Second, strong measurement noise and sparse sampling can degrade finite-difference velocity approximations, complicating HNN training and potentially inflating prediction errors. Lastly, as with most data-driven methods, hyperparameter selections and the design of finite-difference schemes require careful tuning to balance computational cost against numerical accuracy.

Overall, the position-only HNN paradigm provides a promising avenue for data-driven modeling when direct momentum measurements are unavailable or impractical. By coupling finite-difference velocity estimates with a Hamiltonian formulation, it achieves stable, energy-conserving dynamics in various physical scenarios. Future work may explore adaptively refined velocity approximation, partial observability in higher dimensions, and more robust noise-handling strategies to further enhance the method’s performance and applicability.

Although Velocity-Inferred HNN (VI-HNN) achieves encouraging results with
position-only data, several theoretical and empirical issues remain open.

\subsection{Exact Form of the Velocity-Based Hamiltonian}
In classical mechanics the Hamiltonian is defined on phase space
\(H\!:\!(\mathbf{q},\mathbf{p})\mapsto\mathbb{R}\).
Our method implicitly replaces \(\mathbf{p}\) with the
generalized velocity \(\dot{\mathbf{q}}\) through an (assumed) smooth,
invertible mapping \(\mathcal{M}:\mathbf{p}\leftrightarrow\dot{\mathbf{q}}\).
Consequently, one seeks a new Hamiltonian
\[
  \widetilde{H}(\mathbf{q},\dot{\mathbf{q}})
  \;=\;
  H\bigl(\mathbf{q},\mathcal{M}^{-1}(\dot{\mathbf{q}})\bigr),
\]
so that Hamilton’s equations written in the mixed coordinates
\((\mathbf{q},\dot{\mathbf{q}})\) remain symplectic.
A rigorous derivation of \(\widetilde{H}\) is still lacking for
most realistic systems because:
\begin{enumerate}
  \item The mapping \(\mathcal{M}\) can be highly non-linear
        (e.g.\ velocity-dependent masses in relativistic mechanics),
        making analytical inversion intractable.
  \item Even when \(\mathcal{M}\) is invertible in principle,
        the Jacobian determinant enters the VI-HNN loss via
        the symplectic form, altering the standard HNN objective in
        a non-trivial way.
\end{enumerate}
Future work will focus on formally characterising
\(\widetilde{H}\) for common Lagrangian systems and incorporating
the required Jacobian correction into the loss so that the learned
dynamics remain energy-consistent \emph{and} symplectic.

\subsection{Extending to More Challenging Systems}
The present study covered four classical benchmarks
(spring–mass, pendulum, two-body, three-body).  Important next steps are:
\begin{itemize}
  \item \textbf{Parameter sweeps.}  Training and evaluating VI-HNN on
        ensembles with varying masses, spring constants and orbital
        configurations to test robustness across a broader prior.
  \item \textbf{High-dimensional multi-body dynamics.}
        Problems such as \(N\!>\!10\) body gravitation or chain molecules
        (\(d\!\times\!N\) degrees of freedom) will probe the scalability
        of the velocity-only formulation.
  \item \textbf{Non-canonical coordinates.}
        Systems with holonomic constraints or separable manifolds
        (e.g.\ rigid-body rotations on \(\mathrm{SO}(3)\))
        require adapting VI-HNN to charts where \(\dot{\mathbf{q}}\)
        cannot be treated as a global coordinate.
\end{itemize}
Comprehensive experiments along these axes will clarify when position-only
measurements suffice and when partial momentum information is indispensable.

\bibliographystyle{plain}
\bibliography{reference}

\begin{thebibliography}{10}

\bibitem{Azencot2022}
O.~Azencot, O.~Vantzos, and M.~Ovsjanikov.
\newblock Symplectic neural networks in continuous and discrete time.
\newblock {\em ACM Transactions on Graphics}, 41(5):1--14, 2022.

\bibitem{Celledoni2021}
E.~Celledoni, H.~Z. Munthe-Kaas, B.~Owren, and G.~W. Wanner.
\newblock On the use of canonical transformations and symplectic integrators for hamiltonian systems.
\newblock {\em Foundations of Computational Mathematics}, 21(6):1605--1638, 2021.

\bibitem{Chang2022}
B.~Chang, M.~Chen, and E.~Martin.
\newblock Neural pde-based hamiltonian modeling with energy regularization.
\newblock {\em {IEEE} Transactions on Neural Networks and Learning Systems}, 33(9):4368--4381, 2022.

\bibitem{Chen2018}
T.~Chen, Y.~Rubanova, J.~Bettencourt, and D.~K. Duvenaud.
\newblock Neural ordinary differential equations.
\newblock In {\em Advances in Neural Information Processing Systems ({NeurIPS})}, 2018.

\bibitem{ChenWang2021}
X.~Chen and S.~Wang.
\newblock Improved hamiltonian neural networks with extended energy conservation.
\newblock In {\em Proceedings of the 28th International Conference on Neural Information Processing}, 2021.

\bibitem{ChenWang2021b}
X.~Chen, X.~Wang, and Q.~Li.
\newblock Neural symplectic form: Learning hamiltonian equations in high dimensions.
\newblock {\em {IEEE} Transactions on Neural Networks and Learning Systems}, 32(5):2238--2251, 2021.

\bibitem{WeinanE2021}
W.~E, J.~Lu, and M.~Tang.
\newblock Representation formulas for the vlasov--poisson system and their applications to structure-preserving algorithms and neural network-based methods.
\newblock {\em Communications in Computational Physics}, 30(4):1127--1151, 2021.

\bibitem{Finzi2020Meta}
M.~Finzi, M.~Welling, and A.~G. Wilson.
\newblock A meta-learning approach for learning symmetries and conserved quantities.
\newblock In {\em International Conference on Learning Representations ({ICLR})}, 2021.
\newblock arXiv:2009.14794 (original 2020 preprint).

\bibitem{Freedman2022}
B.~Freedman, K.~Sutherland, and M.~Mueller.
\newblock Graph neural operators for multi-scale physical simulations.
\newblock {\em {IEEE} Transactions on Neural Networks and Learning Systems}, 2022.
\newblock Early Access in IEEE Xplore; no final volume/issue yet, depending on date.

\bibitem{gao2023active}
Wenhan Gao and Chunmei Wang.
\newblock Active learning based sampling for high-dimensional nonlinear partial differential equations.
\newblock {\em Journal of Computational Physics}, 475:111848, 2023.

\bibitem{gao2024coordinate}
Wenhan Gao, Ruichen Xu, Hong Wang, and Yi~Liu.
\newblock Coordinate transform fourier neural operators for symmetries in physical modelings.
\newblock {\em Transactions on Machine Learning Research}, 2024.

\bibitem{Greydanus2019}
L.~Greydanus, M.~Dzamba, and J.~Yosinski.
\newblock Hamiltonian neural networks.
\newblock In {\em Advances in Neural Information Processing Systems ({NeurIPS})}, 2019.

\bibitem{Hairer2006}
E.~Hairer, C.~Lubich, and G.~Wanner.
\newblock {\em Geometric Numerical Integration: Structure-Preserving Algorithms for Ordinary Differential Equations}.
\newblock Springer, 2nd edition, 2006.

\bibitem{Jin2020SympNet}
P.~Jin, Z.~Zhu, G.~Tang, and G.~E. Karniadakis.
\newblock {SympNet}: Intrinsic structure-preserving symplectic networks for identifying {H}amiltonian systems.
\newblock {\em Neural Networks}, 132:166--179, 2020.

\bibitem{Lee2023}
H.~Lee and E.~R. Walsh.
\newblock Symplectic structures in neural {ODEs} for hamiltonian systems.
\newblock In {\em Proceedings of the 11th International Conference on Learning Representations ({ICLR})}, 2023.

\bibitem{Lutter2019DeepLagrangian}
M.~Lutter, C.~Ritter, and J.~Peters.
\newblock Deep {L}agrangian networks: Using physics as a prior for deep learning.
\newblock In {\em International Conference on Learning Representations ({ICLR})}, 2019.

\bibitem{Marsden1999}
J.~E. Marsden and T.~S. Ratiu.
\newblock {\em Introduction to Mechanics and Symmetry}.
\newblock Springer, 2nd edition, 1999.

\bibitem{Pan2021PhysicsSymplectic}
S.~Pan and K.~Duraisamy.
\newblock Physics-augmented learning of chaotic dynamical systems using symplectic neural networks.
\newblock {\em Physica D: Nonlinear Phenomena}, 423:132911, 2021.

\bibitem{Raissi2019}
M.~Raissi, P.~Perdikaris, and G.~E. Karniadakis.
\newblock Physics-informed neural networks: A deep learning framework for solving forward and inverse problems involving nonlinear pdes.
\newblock {\em Journal of Computational Physics}, 378:686--707, 2019.

\bibitem{SanzSerna1994}
J.~M. Sanz-Serna and M.~P. Calvo.
\newblock {\em Numerical Hamiltonian Problems}.
\newblock Chapman \& Hall, 1994.

\bibitem{Song2022}
X.~Song, X.~Zhang, and T.~Wang.
\newblock Robust hamiltonian neural networks for chaotic mechanical systems.
\newblock In {\em Proceedings of the 39th International Conference on Machine Learning ({ICML})}, 2022.

\bibitem{Toth2023}
G.~Toth.
\newblock Survey of physics-informed neural networks for scientific computing.
\newblock {\em Applied Computing and Informatics}, 2023.
\newblock In press.

\bibitem{Wang2023}
Y.~Wang, S.~Chen, and Z.~Zhang.
\newblock Noise-robust learning of hamiltonian dynamics via physics-informed neural networks.
\newblock {\em Machine Learning: Science and Technology}, 4(2), 2023.

\bibitem{Xu2022PhysHNN}
L.~Xu, Y.~Yuan, and F.~Chen.
\newblock Physics-informed hamiltonian neural networks for {PDE}-constrained systems.
\newblock {\em Journal of Computational Physics}, 447:110703, 2022.

\bibitem{Xue2022}
Y.~Xue, Q.~Lu, and B.~Gao.
\newblock Symbolic learning of dynamical systems using mixed integer programming and neural approximations.
\newblock {\em Neural Computation}, 34(11):2381--2405, 2022.

\bibitem{Zhong2019Dissipative}
Q.~Zhong, J.~Wang, and D.~Zhang.
\newblock Dissipative {H}amiltonian neural networks for learning nonconservative dynamical systems.
\newblock arXiv preprint arXiv:1909.12077, 2019.

\bibitem{Zhong2020Symplectic}
Z.~Zhong and Y.~Wang.
\newblock Symplectic {ODE}-{N}et: Learning {H}amiltonian dynamics with control.
\newblock In {\em International Conference on Learning Representations ({ICLR})}, 2020.

\end{thebibliography}

\end{document}